\documentclass[acmlarge,nonacm]{acmart}

\usepackage{amsmath,amsthm,amssymb}
\usepackage{booktabs}
\usepackage{xcolor}
\usepackage{colortbl}
\usepackage{pgfplots}
\usepackage{tikz}
\usepackage{tcolorbox}
\pgfplotsset{compat=1.17}
\usepgfplotslibrary{fillbetween}

\newtheorem{theorem}{Theorem}
\newtheorem{corollary}[theorem]{Corollary}

\theoremstyle{definition}
\newtheorem{definition}{Definition}

\theoremstyle{remark}
\newtheorem{remark}{Remark}

\AtBeginDocument{%
  }

\copyrightyear{2026}
\acmYear{2026}
\acmConference[FAccT '26]{The 2026 ACM Conference on Fairness, Accountability, and Transparency}{June 25--28, 2026}{Montreal, QC, Canada}
\acmDOI{10.1145/3805689.3812215}

\begin{document}
\title[Likelihood Ratio Wall]{The Likelihood Ratio Wall: Structural Limits on Accurate Risk Assessment for Rare Violence}

\titlenote{Accepted to the 2026 ACM Conference on Fairness, Accountability, and Transparency (FAccT '26). Please cite the published version: \url{https://doi.org/10.1145/3805689.3812215}}

\author{Marco Pollanen}
\email{marcopollanen@trentu.ca}
\orcid{0000-0001-5356-1889}
\affiliation{%
  \institution{Trent University}
  \city{Peterborough}
  \state{ON}
  \country{Canada}
}

\renewcommand{\shortauthors}{Pollanen}

\begin{abstract}
Pretrial risk assessment tools are used on over one million U.S.\ defendants each year, yet their use for predicting rare violent re-offense faces a basic statistical barrier. We derive a universal precision bound---the \emph{Likelihood Ratio Wall}---showing that when violent re-arrest rates are low (2--5\%), achieving even a 50\% hit rate among people labeled ``high risk'' (positive predictive value, or PPV) would require tools far more discriminative than current instruments appear to be. For rare outcomes, a tool can have respectable-looking performance metrics and still be wrong most of the time it flags someone as ``high risk for violence.'' We show that post-hoc score recalibration cannot solve this problem because it does not improve the tool's underlying ability to separate true positives from false positives. We further prove a \emph{Surveillance Ceiling}: when over-policing inflates recorded ``risk factors'' among those who would not re-offend, the maximum achievable precision is structurally lower for over-policed groups, even at equal offense rates. We translate these results into the \emph{Number Needed to Detain} (how many people must be detained to prevent one violent offense), and propose that risk reports should communicate this uncertainty explicitly. Our findings suggest that for rare violent outcomes, debates about fairness metrics alone are incomplete: under current data regimes, the available features may not support high-confidence individualized detention decisions.
\end{abstract}



\maketitle

\section{Introduction}

Algorithmic risk assessment tools are standard in U.S.\ pretrial proceedings, with over one million defendants evaluated annually \cite{Sawyer2023,DeMichele2018}.
Tools like the Public Safety Assessment (PSA) and COMPAS classify defendants into risk categories, and ``high-risk'' flags influence detention decisions.
The implicit claim is that flagged defendants are substantially more likely to commit violent offenses if released.
This paper examines whether that claim can be sustained mathematically.

The key question is simple: \emph{if a tool labels someone ``high risk,'' how often is that label actually correct?} In statistics, this quantity is called the \textbf{positive predictive value} (PPV). We prove a universal bound: any classifier achieving PPV of at least $\alpha$ when the base rate (outcome frequency) is $\pi$ must have a \textbf{likelihood ratio} (LR), measuring how much more often the tool flags true positives than false positives, satisfying

\[
\text{LR} = \frac{\Pr(\text{flag} \mid \text{positive})}{\Pr(\text{flag} \mid \text{negative})} \geq \frac{\alpha(1-\pi)}{(1-\alpha)\pi}.
\]
 The ``hit rate'' of a high-risk flag thus depends on two things: how rare the outcome is and how strongly the flag distinguishes true positives from false positives.

\paragraph{Definitions and scope.}
We focus on \emph{violent pretrial re-arrest}, defined as arrest for a violent felony offense within the pretrial release period (typically 6--24 months). Base rates range from 2\% to 5\% in most jurisdictions \cite{DeMichele2018,Lowenkamp2013}; in reform-oriented jurisdictions like New York City, rates are often below 2\% \cite{NYC2021}. We distinguish this from broader outcomes such as ``any re-arrest'' (base rates 20--30\%), for which the precision constraints we identify are substantially less severe. This distinction matters because many published evaluations report performance for more common outcomes, which can make tools appear more useful for detention decisions than they are for rare violent outcomes.

At a 3\% base rate, a tool needs $\text{LR} \geq 32$ for a ``high-risk'' violence flag to be correct even half the time. Published evidence suggests current tools achieve $\text{LR} \approx 2\text{--}6$ at operational thresholds (we substantiate this in Section~6). A 50\% threshold corresponds to what lawyers call the \textbf{preponderance of the evidence} standard, meaning, roughly, ``more likely than not.'' We use that benchmark as an analogy, not as a claim about current constitutional doctrine \cite{Addington1979}; we discuss the justification for this choice, and why it is conservative, in Section~8.2.
This gap is structural: recalibration preserves likelihood ratios and cannot close it.
Differential surveillance compounds the problem: we prove a \emph{Surveillance Ceiling} showing that, under a stylized count-threshold model, over-policing degrades the maximum achievable precision for affected populations. Beyond the theoretical results, we propose a practical policy response: risk reports used in detention contexts should display plain-language uncertainty information (Section~\ref{sec:transparency}).

\paragraph{Why PPV and LR matter.}
PPV answers the question judges most naturally care about: ``if this person is flagged high risk, how likely is it the flag is correct?'' LR is the bridge between the tool's signal quality and that hit rate: once the base rate is fixed, LR fully determines PPV. By contrast, the widely reported Area Under the ROC Curve (AUC) measures how well a tool \emph{ranks} people from lower to higher risk, but does not directly indicate how often a ``high-risk'' flag is correct (Section~9.3).

\paragraph{Contributions.}
We make four contributions:
\begin{enumerate}
\item \textbf{A base-rate-driven precision limit.} We formalize the Likelihood Ratio Wall (Theorem~\ref{thm:wall}), quantifying the gap between required and achieved likelihood ratios for rare violence prediction.
\item \textbf{Why recalibration cannot fix this limit.} We prove recalibration invariance (Theorem~\ref{thm:recal}): monotone score transformations preserve likelihood ratios at every threshold. This property is known in signal detection theory; our contribution is identifying it as a barrier to post-hoc fixes in pretrial risk assessment.
\item \textbf{A formal model of how over-surveillance lowers achievable precision.} We introduce the Surveillance Ceiling (Theorem~\ref{thm:ceiling}), demonstrating under stylized assumptions how asymmetric feature accumulation degrades maximum precision for over-policed populations even at equal base rates.
\item \textbf{Policy-facing metrics and transparency recommendations.} We translate the results into Number Needed to Detain (NND) and propose uncertainty disclosures for risk reports.
\end{enumerate}

The underlying probability relationship is standard in epidemiology \cite{Altman1994,Ioannidis2005}. Our novelty lies in applying it to pretrial violence prediction: focusing on ``high-risk'' cutoffs that trigger action, quantifying what discrimination is needed, and proving the Surveillance Ceiling as a second structural barrier independent of algorithmic design. Much prior work asks whether a tool treats groups differently. We identify a different problem: in low-prevalence settings, the available data may not contain enough signal for accurate individualized detention decisions. Unequal surveillance can further weaken that signal by increasing false alarms in the recorded data.

\paragraph{Paper structure.}
Section~2 reviews related work. Section~3 establishes the Likelihood Ratio Wall.
Section~4 proves recalibration invariance. Section~5 develops the Surveillance Ceiling. Section~6 presents empirical validation.
Section~7 addresses the counterfactual objection. Section~8 develops policy implications. Section~9 discusses broader implications. Sections~10--12 cover limitations, ethics, and conclusions.
Readers primarily interested in policy may focus on Sections~1, 3, 6, and~8.

\begin{tcolorbox}[colback=gray!5!white,colframe=gray!75!black,title=\textbf{Key Terms Used in This Paper}]
\small
\begin{itemize} {
    \item \textbf{Base rate} ($\pi$): how common the outcome is (e.g., the violent re-arrest rate).
    \item \textbf{Positive predictive value (PPV)}: among people flagged ``high risk,'' the share who actually experience the outcome. Also called \emph{precision}.
    \item \textbf{Likelihood ratio (LR)}: how much more likely the tool is to flag someone who will have the outcome than someone who will not. Higher LR means stronger signal.
    \item \textbf{AUC} (Area Under the ROC Curve): a ranking metric that measures how well a tool ranks higher-risk people above lower-risk people, but \emph{not} a direct measure of how often ``high-risk'' flags are correct.
    \item \textbf{Number Needed to Detain (NND)}: how many flagged defendants must be detained to prevent one violent offense ($= 1/\text{PPV}$).}
\end{itemize}
\end{tcolorbox}

\section{Related Work}

Our work bridges algorithmic fairness, clinical epidemiology, and legal theory.

\subsection{Impossibility Results in Algorithmic Fairness}

Prior fairness research asks whether a model can satisfy multiple desirable fairness criteria simultaneously. We address a different question: whether the model can be sufficiently accurate for rare violent outcomes in the first place. Chouldechova \cite{Chouldechova2017} proved that when base rates differ between groups, a classifier cannot simultaneously satisfy calibration (equal PPV across groups) and balance (equal false positive and false negative rates across groups).
Kleinberg et al.\ \cite{Kleinberg2017} extended this to show inherent trade-offs among multiple fairness definitions. We contribute a different impossibility: regardless of fairness definition, \emph{useful} PPV is unattainable for rare outcomes given the discriminative power of available features. The Likelihood Ratio Wall is a constraint on information, not equity.

\subsection{Diagnostic Accuracy and Base Rates}

A familiar problem in medicine is that when a condition is rare, even a good test can produce many false alarms. We apply the same logic to ``high-risk'' violence flags.
Altman and Bland \cite{Altman1994} demonstrated that even sensitive and specific tests yield poor PPV when screening for rare conditions.
Ioannidis \cite{Ioannidis2005} extended this to argue that many published findings are false because they test low-prior-probability hypotheses.
Major pretrial risk assessment validation studies \cite{DeMichele2018,Lowenkamp2013} typically report AUC but rarely report PPV at operational thresholds, obscuring the precision deficit we identify. Stevenson \cite{Stevenson2018} found that PSA implementation had minimal effect on detention rates;
our analysis suggests judges may implicitly recognize that ``high-risk'' flags carry limited informational weight.

\subsection{Feedback Loops and Surveillance}

Existing work shows that unequal surveillance can reproduce or amplify disparities. We extend this by showing a related result: unequal surveillance can also lower the best precision any classifier can achieve from those features.
Lum and Isaac \cite{Lum2016} demonstrated that predictive policing trained on arrest data reproduces over-policing patterns. Ensign et al.\ \cite{Ensign2018} formalized conditions under which such loops amplify disparities.
Eubanks \cite{Eubanks2018} documents how algorithmic systems entrench poverty by treating symptoms of deprivation as risk factors;
Benjamin \cite{Benjamin2019} argues that ``risk'' categories sanitize racialized surveillance. Our Surveillance Ceiling provides a formal foundation for these critiques, extending the literature from bias to capacity.

\subsection{Legal Standards of Proof}

In U.S.\ legal terminology, \textbf{preponderance of the evidence} means ``more likely than not'' ($>50\%$ probability), while \textbf{clear and convincing evidence} denotes substantially higher confidence. The Supreme Court in \emph{Addington v.\ Texas} applied these concepts in civil commitment \cite{Addington1979}. Pretrial detention involves a distinct legal framework: standards vary by jurisdiction, and dangerousness determinations are often not expressed as explicit probabilistic thresholds. Nevertheless, detention is liberty-depriving, and statistical confidence remains normatively relevant, a point emphasized by legal scholars who criticize actuarial tools for projecting confidence that may exceed evidentiary support at the individual level \cite{Harcourt2007,Starr2014}. Skeem and Monahan \cite{Skeem2011} noted the limited clinical utility of violence risk assessment.
Mayson \cite{Mayson2018} posed the foundational question our analysis quantifies: what probability of future crime justifies pretrial detention? She showed that this question remains unresolved in existing doctrine, with no consensus on the required risk threshold.
We use these legal labels as \emph{analogical benchmarks}, not claims that detention law requires a particular PPV threshold. Our point is comparative: the confidence implied by ``high-risk'' labels is often far lower than readers may assume.

\section{The Likelihood Ratio Wall}

This section states the core mathematical result and translates it: how strong would a tool need to be for a ``high-risk'' violence flag to be right at least half the time? The bound is architecture-independent: it applies to logistic regression, random forests, neural networks, and human judgment alike.

\subsection{Preliminaries}

We consider a binary classification setting where $Y \in \{0,1\}$ denotes the true outcome (think of $Y=1$ as ``a violent re-arrest occurs if released'') and $F \in \{0,1\}$ denotes the classifier's flag ($F=1$ meaning ``the tool labels the person high risk''). Let $\pi = \Pr(Y=1)$ denote the base rate. We define:

\begin{itemize}
    \item \textbf{Sensitivity} (true positive rate): $s = \Pr(F=1 \mid Y=1)$
    \item \textbf{False positive rate}: $q = \Pr(F=1 \mid Y=0)$
    \item \textbf{Likelihood ratio}: $\text{LR} = s/q$
    \item \textbf{Positive predictive value} (PPV): $\Pr(Y=1 \mid F=1)$, the share of flagged cases that are true positives
\end{itemize}

\subsection{The Universal Bound}

The following theorem formalizes a basic but often hidden fact: when the outcome is rare, the tool's flag must be extremely selective to be correct even half the time.

\begin{theorem}[Universal Precision Bound --- The Likelihood Ratio Wall]
\label{thm:wall}
For any binary classifier with likelihood ratio $\text{LR} = s/q$ operating at base rate $\pi$:
\[
\text{PPV} = \frac{1}{1 + \frac{1-\pi}{\pi} \cdot \frac{1}{\text{LR}}}.
\]
Equivalently, achieving $\text{PPV} \geq \alpha$ requires:
\begin{equation}
\text{LR} \geq \frac{\alpha}{1-\alpha} \cdot \frac{1-\pi}{\pi}.
\label{eq:wall}
\end{equation}
\end{theorem}

\begin{proof}

By Bayes' theorem:
$\text{PPV} = \frac{s \pi}{s \pi + q(1-\pi)}$.
Dividing by $q(1-\pi)$:
$\text{PPV} = \frac{\text{LR} \cdot \Omega}{1 + \text{LR} \cdot \Omega}$
where $\Omega = \pi/(1-\pi)$ is the prior odds, giving the first formula. Setting $\text{PPV} \geq \alpha$ and solving for LR yields Equation~\eqref{eq:wall}.

\end{proof}

In plain terms: the rarer the event, the stronger the signal must be for a ``high-risk'' flag to be trustworthy. Required LRs grow hyperbolically as base rates decrease.

\begin{corollary}[The 99-to-1 Rule: a worked example]
\label{cor:99to1}
At $\pi = 1\%$, achieving $\text{PPV} \geq 50\%$ requires $\text{LR} \geq 99$---the tool must be roughly 99 times more likely to flag a true positive than a false positive. This follows directly from substituting into Equation~\eqref{eq:wall}.
\end{corollary}

Table~\ref{tab:wall} displays required likelihood ratios across base rates and PPV targets.

\begin{table}[t]
\centering

\caption{Required LR (how much more often the tool must flag true positives than false positives) to achieve a target PPV (share of ``high-risk'' flags that are correct).
Shaded cells require LR values exceeding current instrument performance ($\text{LR} \approx 2\text{--}6$; see Section~6).
At typical violent re-arrest rates (2--5\%), even 50\% PPV requires LR well beyond this range.}

\label{tab:wall}
\begin{tabular}{l|cccc}
\toprule
Base rate $\pi$ & \multicolumn{4}{c}{Required LR for target PPV} \\
 & 25\% & 50\% & 75\% & 90\% \\
\midrule
10\% & 3.0 & \cellcolor{gray!15}9.0 & \cellcolor{gray!15}27 & \cellcolor{gray!15}81 \\
5\% & \cellcolor{gray!15}6.3 & \cellcolor{gray!15}19 & \cellcolor{gray!15}57 & \cellcolor{gray!15}171 \\
3\% & \cellcolor{gray!15}10.8 & \cellcolor{gray!15}32.3 & \cellcolor{gray!15}97 & \cellcolor{gray!15}291 \\
2\% & \cellcolor{gray!15}16.3 & \cellcolor{gray!15}49 & \cellcolor{gray!15}147 & \cellcolor{gray!15}441 \\
1\% & \cellcolor{gray!15}33 & \cellcolor{gray!15}99 & \cellcolor{gray!15}297 & \cellcolor{gray!15}891 \\
\bottomrule
\end{tabular}
\end{table}

\subsection{Visualizing the Wall}

Figure~\ref{fig:wall} plots curves showing combinations of base rate and likelihood ratio that produce the same PPV. The gray-shaded region below the solid black curve represents settings where PPV falls below 50\%, where the ``high-risk'' flag is wrong more often than right. At base rates typical of pretrial violent re-arrest (2--5\%), tools operating in the LR range reported for COMPAS fall entirely within this region.

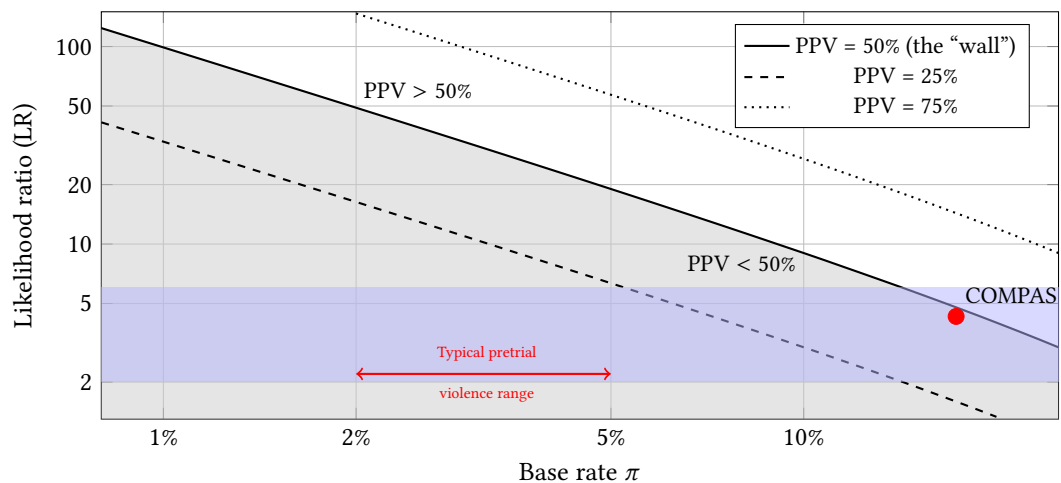
\begin{figure}[t]
\centering
\begin{tikzpicture}
\begin{axis}[
    width=0.9\columnwidth,
    height=7cm,
    xlabel={Base rate $\pi$},
    ylabel={Likelihood ratio (LR)},
    xmode=log,
    ymode=log,
    xmin=0.008, xmax=0.25,
    ymin=1.3, ymax=150,
    xtick={0.01,0.02,0.05,0.1},
    xticklabels={1\%,2\%,5\%,10\%},
    ytick={2,5,10,20,50,100},
    yticklabels={2,5,10,20,50,100},
    grid=major,
    legend pos=north east,
    legend style={font=\small},
    clip=false
]

\addplot[thick, black, domain=0.008:0.25, samples=50, name path=wall] {(1-x)/x};
\addlegendentry{PPV = 50\% (the ``wall'')}

\path[name path=bottom] (axis cs:0.008,1.3) -- (axis cs:0.25,1.3);
\addplot[gray!40, opacity=0.5, forget plot] fill between[of=bottom and wall];

\addplot[thick, black, dashed, domain=0.008:0.20, samples=50] {(1-x)/(3*x)};
\addlegendentry{PPV = 25\%}

\addplot[thick, black, dotted, domain=0.02:0.25, samples=50] {3*(1-x)/x};
\addlegendentry{PPV = 75\%}

\fill[blue!30, opacity=0.5] (axis cs:0.008,2) rectangle (axis cs:0.25,6);

\addplot[only marks, mark=*, mark size=3pt, red] coordinates {(0.173, 4.3)};
\node[font=\small, anchor=south west] at (axis cs:0.173, 4.5) {COMPAS};

\draw[thick, red, <->] (axis cs:0.02, 2.2) -- (axis cs:0.05, 2.2);
\node[font=\tiny, anchor=south, red] at (axis cs:0.032, 2.3) {Typical pretrial};
\node[font=\tiny, anchor=north, red] at (axis cs:0.032, 2.1) {violence range};

\node[font=\small] at (axis cs:0.025, 60) {PPV $> 50\%$};
\node[font=\small] at (axis cs:0.08, 8.0) {PPV $< 50\%$};

\end{axis}
\end{tikzpicture}

\caption{The Likelihood Ratio Wall.
Curves show the LR required for each PPV level. The gray region below the solid curve is where the ``high-risk'' flag is wrong more often than right (PPV $< 50\%$). The blue band shows approximate LR range ($\approx 2\text{--}6$) from validation literature (Section~6). The red COMPAS point is plotted at the Broward County base rate ($\pi = 17.3\%$) and estimated LR ($4.3$; Table~\ref{tab:compas}). The red bracket marks the typical pretrial violent re-arrest range (2--5\%), where instruments in the blue band fall squarely in the shaded region.}
\Description[Log-log plot showing the likelihood ratio required for various PPV thresholds across base rates from 1\% to 25\%, with current instrument performance falling well short of what is needed for rare violent re-arrest.]{Log-log plot of likelihood ratio (vertical axis, 1.3 to 150) against base rate (horizontal axis, 1\% to 25\%). Three downward-sloping curves show the LR required for PPV equal to 25\% (dashed), 50\% (solid), and 75\% (dotted); all three curves rise steeply as the base rate decreases. A gray-shaded region lies below the 50\% solid curve, marking combinations where the high-risk flag is wrong more often than right. A horizontal blue band between LR 2 and 6 shows the empirical performance range of current instruments. A red dot at base rate 17.3\%, LR 4.3, is labelled COMPAS. A red bracket along the lower portion of the plot, between base rates 2\% and 5\%, marks the typical pretrial violent re-arrest range; instruments operating in the blue band within this base-rate range fall entirely within the gray-shaded region where PPV is below 50\%.}
\label{fig:wall}
\end{figure}

\section{Recalibration Invariance}

A common response to disappointing model performance is to \emph{recalibrate}: rescale scores so that predicted probabilities match observed frequencies. This can improve how well probabilities are stated, but does not improve the tool's ability to separate positives from negatives. Standard recalibration methods, including Platt scaling \cite{Platt1999}, isotonic regression \cite{Zadrozny2002}, and related approaches \cite{NiculescuMizil2005}, produce non-decreasing transformations of the original score, preserving the ordering of individuals by risk. The invariance of operating points under monotone transforms is known in signal detection theory; our contribution is identifying it as a structural barrier to post-hoc fixes in pretrial risk assessment.

\begin{theorem}[Recalibration Preserves Likelihood Ratios]
\label{thm:recal}
Let $S$ be a continuous risk score and $f: \mathbb{R} \to \mathbb{R}$ be any strictly increasing transformation.
For every threshold $t$:
\[
\text{LR}_{f(S)}(f(t)) = \text{LR}_S(t).
\]
\end{theorem}

\begin{proof}

Since $f$ is strictly increasing, for any $t$ the event $\{S \geq t\}$ is identical to $\{f(S) \geq f(t)\}$: if $S \geq t$ then $f(S) \geq f(t)$ by monotonicity, and if $S < t$ then $f(S) < f(t)$ by strict increase. The two events have the same indicator function, so conditioning on any $\sigma$-algebra (including $\{Y=1\}$ or $\{Y=0\}$) preserves equality:
\begin{align*}
s_{f(S)}(f(t)) &= \Pr(f(S) \geq f(t) \mid Y=1) = \Pr(S \geq t \mid Y=1) = s_S(t), \\
q_{f(S)}(f(t)) &= \Pr(f(S) \geq f(t) \mid Y=0) = \Pr(S \geq t \mid Y=0) = q_S(t).
\end{align*}
Both sensitivity and false positive rate are preserved at the corresponding threshold, so their ratio is invariant:
\[
\text{LR}_{f(S)}(f(t)) = \frac{s_{f(S)}(f(t))}{q_{f(S)}(f(t))} = \frac{s_S(t)}{q_S(t)} = \text{LR}_S(t). \qedhere
\]

\end{proof}

\begin{remark}
If $f$ is only non-decreasing (e.g., isotonic regression), it preserves LRs except at thresholds intersecting flat regions. Even more flexible recalibration methods cannot create performance beyond the best trade-offs already present in the original score.
\end{remark}

In practical terms, changing the score scale while preserving ordering does not create new predictive signal for threshold-based detention decisions. Recalibration improves the \emph{reliability} of probability estimates but not \emph{discrimination}.
If the underlying features achieve only modest separation (for example, $\text{LR} \approx 3$, which at low base rates yields very poor PPV), no monotonic transformation can improve it.

\section{The Surveillance Ceiling}

Communities subject to differential policing accumulate observable ``risk factors''
through increased exposure to the criminal legal system rather than elevated
underlying risk. The goal of this section is to isolate one mechanism by which this unequal surveillance can increase false alarms and lower achievable precision. We derive a \emph{mechanism-isolating directional result} for a class of score constructions, demonstrating how asymmetric data generation can induce different upper limits on achievable PPV across groups.

\subsection{Setup, Assumptions, and Practical Relevance}

We use a simplified model with repeated binary ``risk markers'' to make the mechanism transparent. Many pretrial risk tools use additive point-score constructions that count accumulated administrative indicators: prior convictions, arrests, failure-to-appear counts, pending charges \cite{DeMichele2018,Lowenkamp2013}. Count-threshold classifiers capture the core logic of such instruments: an individual is flagged when enough indicators are present. Our model does not capture nonlinear interactions or richer covariates; it speaks most directly to score components built from accumulated administrative markers.

Consider two populations, $A$ and $B$, with binary outcome $Y\in\{0,1\}$ and $k$ binary risk factors $X_1,\dots,X_k$. We assume:
\begin{itemize}
    \item \textbf{Equal base rates}: $\Pr(Y=1\mid G=A)=\Pr(Y=1\mid G=B)=\pi$.
    \item \textbf{Independence} (for tractability): the $X_i$ are i.i.d.\ given $(Y,G)$.
    \item \textbf{Identical positive-class distributions}: $X_i \mid (Y=1,G) \sim \mathrm{Bernoulli}(p_+)$ for both groups.
    \item \textbf{Different negative-class prevalences}: $X_i \mid (Y=0,G) \sim \mathrm{Bernoulli}(p_G)$, with $p_A < p_B$.
\end{itemize}

The inequality $p_A<p_B$ captures the mechanism documented by Lum and Isaac \cite{Lum2016}: over-policed communities accumulate administrative markers at higher rates even among those who do not commit violent offenses. We study \emph{count-threshold classifiers} that flag an individual when $\sum_{i=1}^k X_i \ge m(k)$.

\subsection{The Ceiling Theorem}

The following theorem states two points: (1) when recorded risk markers are more common among negatives in one group, false positive rates grow rapidly as more markers are used; (2) this lowers the best achievable LR, and therefore the best achievable PPV (via Theorem~\ref{thm:wall}), for the over-policed group.

\begin{theorem}[Surveillance Ceiling for Count-Threshold Classifiers]
\label{thm:ceiling}
Fix a threshold sequence
$m(k)=\lceil k\theta\rceil$
for some $\theta\in(0,1)$ satisfying $\theta>p_B$.
Let $q_G(k)$ denote the false positive rate in group $G$ and $s(k)$ the sensitivity (common across groups).
Assume the sample mean $S_k = \frac{1}{k}\sum X_i$ satisfies a \emph{Large Deviation Principle} (LDP), a standard regularity condition describing how quickly tail probabilities decay, with rate function $I_G(x)$ for the negative class in each group.

Then:
\begin{enumerate}
\item[(i)] The false positive rate ratio grows exponentially in $k$:
\[
\lim_{k\to\infty}\frac{1}{k}\log\frac{q_B(k)}{q_A(k)}
= I_A(\theta)-I_B(\theta) \;>\; 0.
\]
\item[(ii)] $\mathrm{LR}_B(k) < \mathrm{LR}_A(k)$ for every threshold, and
$\sup_m \mathrm{LR}_B(m) \le \sup_m \mathrm{LR}_A(m)$
(where the supremum is the best achievable value over count-threshold classifiers). Because PPV is strictly increasing in LR for fixed base rate (Theorem~\ref{thm:wall}), this LR ceiling implies a \emph{PPV ceiling}: maximum achievable PPV is lower in the over-policed group.
\end{enumerate}
\end{theorem}

\begin{proof}[Proof sketch]
The false positive rate depends on how often negative cases exceed the threshold. If surveillance increases marker prevalence among negatives in Group~$B$, then negative cases in $B$ cross the threshold more often. Under the LDP, this probability decays at rate $I_G(\theta)$, and $p_B > p_A$ implies $I_B(\theta) < I_A(\theta)$ (thicker tail in $B$). Under i.i.d.\ Bernoulli, $I_G(\theta) = D(\theta \| p_G)$ (the KL divergence), and $D(\theta \| p_A) > D(\theta \| p_B)$ since $\theta > p_B > p_A$. Since sensitivity is identical across groups, higher FPR in $B$ implies lower LR and thus lower PPV. Full proof: Appendix~\ref{app:ceiling_proof}.
\end{proof}

\paragraph{Interpreting the theorem in practice.}
The identical positive-class assumption holds sensitivity constant to isolate the effect of surveillance-driven negative-class inflation. In practice, positive-class distributions may differ: if positive-class marker prevalence is also higher in the more-policed group (e.g., because detected offenders accumulate more system contact), the sensitivity advantage may partially offset the FPR penalty, narrowing the gap. Conversely, if positive-class prevalence is similar while negative-class prevalence diverges, the theorem's prediction is conservative. The theorem supports a \emph{directional conclusion}: surveillance-induced inflation of negative-class feature prevalence lowers LR and PPV for affected groups. It is a comparative benchmark, not a literal estimator of deployed-system group differences.

\paragraph{Robustness to assumption violations.}
When features are correlated rather than independent, the directional intuition is preserved: higher feature prevalence among negatives still produces higher false positive rates. If the classifier uses weighted combinations rather than count thresholds, the quantitative predictions do not apply directly, but any model influenced by aggregate feature prevalence will exhibit qualitatively similar behavior. If the surveillance effect is concentrated in a subset of features, the ceiling is attenuated but remains.

To illustrate the sensitivity to the identical positive-class assumption, consider a variant of Table~\ref{tab:correlated_ceiling} in which the positive-class marker prevalence is 10\% higher in the more-policed group ($p_+^B = 0.75$ vs.\ $p_+^A = 0.68$), reflecting greater system contact among detected offenders. In this case, sensitivity rises from 62.4\% to approximately 72\% in Group~$B$, partially offsetting the FPR penalty: the LR gap narrows from roughly 13-fold to approximately 12-fold, and the PPV gap narrows from roughly 6-fold to approximately 5-fold. The directional ceiling remains (Group~$B$ still has strictly lower LR and PPV than Group~$A$), but its magnitude is attenuated. Conversely, if positive-class prevalence is equal or lower in Group~$B$ while negative-class prevalence diverges, the theorem's prediction is conservative. Empirical work comparing achievable LRs across groups with different enforcement exposure would be valuable for quantifying these effects in deployed systems.

\subsection{Finite-Sample Illustration}

Figure~\ref{fig:ceiling} illustrates the mechanism schematically. Table~\ref{tab:correlated_ceiling} reports a  deterministic numerical example (not a Monte Carlo simulation) with $k=10$ binary risk factors and moderate correlation among the risk-factor indicators ($\rho=0.2$). Despite identical base rates ($\pi = 3\%$) and identical positive-class behavior, the same classifier achieves PPV above 50\% in the less-policed group but collapses to 11\% in the more-policed group due solely to inflated negative-class prevalence.

\begin{table}[t]
\centering

\caption{Surveillance Ceiling illustration with correlated features ($\rho=0.2$ among binary risk-factor indicators). Values are computed from specified parameters, not estimated from repeated trials. Despite identical base rates and positive-class distributions, surveillance produces a 6-fold PPV gap. NND = $1/\text{PPV}$: flagged defendants detained per prevented offense.}

\label{tab:correlated_ceiling}
\begin{tabular}{lcc}
\toprule
Metric & Group A (Less Policed) & Group B (More Policed) \\
\midrule
Neg.-Class Prevalence ($p_G$) & 15\% & 35\% \\
False Positive Rate ($q$) & 1.1\% & 14.8\% \\
Sensitivity ($s$) & 62.4\% & 62.4\% \\
\midrule
\textbf{Likelihood Ratio (LR)} & \textbf{56.7} & \textbf{4.2} \\
\textbf{Projected PPV (at $\pi=3\%$)} & \textbf{64\%} & \textbf{11\%} \\
\textbf{Number Needed to Detain (NND)} & \textbf{1.6} & \textbf{9.1} \\
\bottomrule
\end{tabular}
\end{table}

\begin{figure}[t]
\centering
\begin{tikzpicture}
\begin{axis}[
    width=0.9\columnwidth,
    height=6cm,
    xlabel={Risk Score},
    ylabel={Density},
    domain=0:10,
    samples=100,
    axis lines=left,
    ytick=\empty,
    xtick={5},
    xticklabels={Threshold},
    legend style={at={(0.97,0.97)},anchor=north east, font=\small},
]
\addplot [thick, red, fill=red!15, opacity=0.6] {exp(-(x-7)^2/2)};
\addlegendentry{Positives (both)}

\addplot [thick, blue, dashed] {exp(-(x-2)^2/2.5)};
\addlegendentry{Negatives (A)}

\addplot [thick, blue, densely dotted, fill=blue!15, opacity=0.4] {exp(-(x-4)^2/2.5)};
\addlegendentry{Negatives (B)}

\draw [thick, black] (axis cs:5,0) -- (axis cs:5,1.15);
\draw[->, very thick, black] (axis cs:2.3, 0.75) -- (axis cs:3.7, 0.75);
\node[font=\small, anchor=south] at (axis cs:3, 0.75) {Surveillance shift};

\end{axis}
\end{tikzpicture}

\caption{Stylized schematic (not empirical density estimates) of the Surveillance Ceiling mechanism.
Over-policing shifts the negative-class distribution rightward for Group~B (dotted blue), increasing overlap with positives and inflating FPR. The positive class is identical across groups.}
\Description[Schematic of three overlapping density curves illustrating how surveillance shifts the negative-class distribution toward the threshold for the over-policed group.]{Stylized density plot with three curves above a horizontal axis labelled Risk Score, with a vertical line marking the threshold near the centre. A solid red curve labelled ``Positives (both)'' peaks well to the right of the threshold. A dashed blue curve labelled ``Negatives (A)'' peaks well to the left of the threshold, with little overlap with the positive curve. A dotted blue curve labelled ``Negatives (B)'' peaks between the two, with substantial overlap with the positive curve. A horizontal arrow labelled ``Surveillance shift'' points from the Group A peak toward the Group B peak, illustrating that over-policing shifts the negative-class distribution rightward, pushing more of the negative cases above the threshold.}
\label{fig:ceiling}
\end{figure}
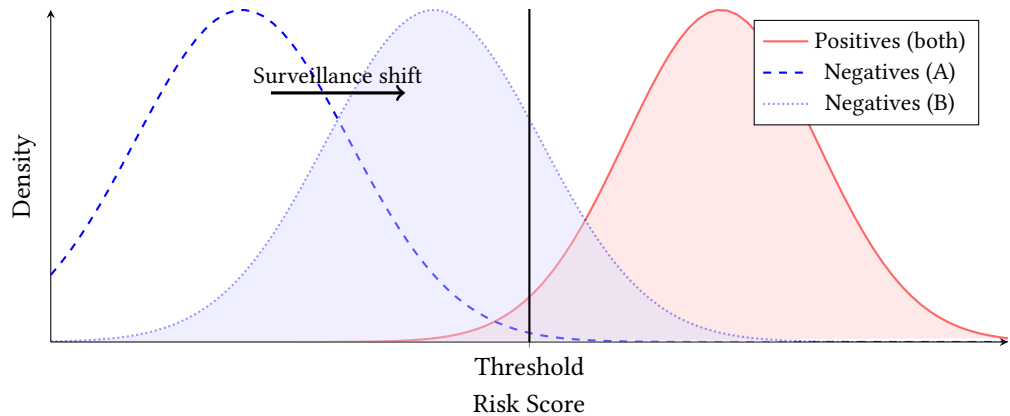

\section{Empirical Validation}

This section illustrates the paper's key quantities on a public dataset, not to claim that one historical dataset represents all modern deployments. Due to the proprietary nature of many modern tools, comprehensive multi-jurisdictional data with outcomes is often unavailable.
We use the ProPublica COMPAS dataset ($N = 4{,}743$) from Broward County, Florida \cite{Angwin2016} as an illustrative example. COMPAS remains one of the few publicly available datasets permitting direct LR calculation from individual-level scores and outcomes.
Previous research suggests that simple linear models of criminal history perform similarly across jurisdictions \cite{Rudin2019}.

\paragraph{Validation via meta-analysis.}
Many validation studies report AUC (how well a tool ranks people) but not how often a ``high-risk'' label is correct. A systematic review by Desmarais, Johnson, and Singh found a median AUC of 0.68 for predicting violent recidivism across 19 instruments \cite{Desmarais2013}. An AUC of 0.70 generally corresponds to LRs in the low single digits ($< 5$) at operational thresholds \cite{Hand2009}.

\paragraph{Validation of modern instruments.}
Even the PSA, the most widely adopted modern instrument, demonstrates similar limitations. A statewide validation funded by its developer reported AUCs of 0.64--0.66 for violent recidivism \cite{DeMichele2018}. At those levels, pushing PPV above 50\% at a 3\% base rate would require such an extreme cutoff that the tool would flag very few people and miss most true positives.

\subsection{Likelihood Ratio Estimation}

Table~\ref{tab:compas} reports COMPAS performance for violent recidivism. The 17.3\% base rate reflects a broader definition than typical pretrial violent felony re-arrest (2--5\%). Even at this elevated base rate, PPV barely reaches the ``more likely than not'' benchmark.

\begin{table}[t]
\centering

\caption{COMPAS ``high risk'' flag performance for violent recidivism ($\pi = 17.3\%$). PPV = share of flagged cases that are true positives.}

\label{tab:compas}
\begin{tabular}{lcc}
\toprule
Metric & Estimate & 95\% CI \\
\midrule
Sensitivity ($s$) & 39.6\% & (34.2\%, 45.1\%) \\
False Positive Rate ($q$) & 9.2\% & (8.1\%, 10.4\%) \\
Likelihood Ratio (LR) & 4.3 & (3.1, 5.8) \\
PPV at $\pi = 17.3\%$ & 47\% & (41\%, 53\%) \\
\bottomrule
\end{tabular}
\end{table}

Table~\ref{tab:projected} projects the same LR to realistic pretrial base rates. As base rates drop, PPV collapses.

\begin{table}[t]
\centering

\caption{Projected PPV and NND (defendants detained per prevented offense) at typical pretrial base rates, assuming $\text{LR} = 4.3$.}

\label{tab:projected}
\begin{tabular}{lcc}
\toprule
Base Rate & Projected PPV & NND \\
\midrule
17.3\% (Broward) & 47\% & 2.1 \\
5\% & 18\% & 5.5 \\
3\% & 12\% & 8.5 \\
2\% & 8\% & 12.5 \\
\bottomrule
\end{tabular}
\end{table}

\subsection{Surveillance Ceiling Validation}

We test Theorem~\ref{thm:ceiling}'s prediction that false positive rates should increase by a larger factor than average factor counts (\emph{superlinear amplification}) using COMPAS data stratified by race. Table~\ref{tab:fpr_amplification} shows that a $1.8\times$ ratio in mean factor count produces a $2.9\times$ ratio in false positive rates, consistent with the theorem's prediction. This is an illustrative empirical check, not a causal identification of surveillance effects.

\begin{table}[t]
\centering

\caption{FPR amplification in COMPAS violent recidivism predictions. The FPR ratio ($2.9\times$) exceeds the mean factor ratio ($1.8\times$), illustrating superlinear amplification.}

\label{tab:fpr_amplification}
\begin{tabular}{lccc}
\toprule
Group & Mean Risk Factors & FPR (Violent) & Ratio \\
\midrule
White defendants & 1.96 & 4.9\% & 1.0$\times$ \\
Black defendants & 3.54 & 14.2\% & 2.9$\times$ \\
\bottomrule
\end{tabular}
\end{table}

\section{The Counterfactual Problem}

\textbf{Objection:} We estimate risk using observed outcomes among released defendants. If detention prevents some offenses, the ``true'' rate could be higher, reducing the required LR. This is the strongest objection to our analysis. We offer three responses.

\paragraph{1. The estimand is prospective.}
The decision-relevant quantity for a judge is the probability of violence given release and the information available at the decision point. This is precisely what observed re-arrest rates among released defendants estimate. Applying the tool to a broader population requires transportability assumptions. The released population is plausibly lower-risk (since judges selectively detain the most dangerous), meaning the true population base rate may be \emph{higher} than observed. By Theorem~\ref{thm:wall}, a higher base rate \emph{lowers} the required LR, making our analysis more conservative.

\paragraph{2. Incapacitation effects are modest.}
Dobbie, Goldin, and Yang \cite{Dobbie2018}, exploiting quasi-random judge assignment, found releasing marginal detainees increased violent crime by about 0.5 percentage points.

\begin{tcolorbox}[colback=gray!5!white,colframe=gray!60!black,width=\columnwidth]
\small
\textbf{Robustness check:} Even doubling the base rate from 3\% to 6\% reduces the required LR from 32 to approximately 16, still well above the empirical range of 2--6.
\end{tcolorbox}

{
\paragraph{3. A dilemma for tool proponents.}
If detained defendants are vastly more dangerous than released defendants, judges are already successfully identifying high-risk individuals \emph{without} the algorithm.
If the algorithm cannot outperform judicial intuition on the released population (where outcomes are observable), why trust its predictions on the detained population (where they are not)?
Either the tool is informationally redundant given the LR wall, or it is inaccurate because violent re-offense is rare among those it is applied to.}

\section{Policy Implications}

\subsection{Number Needed to Detain}

We adapt the Number Needed to Treat (NNT) from clinical epidemiology \cite{Altman1994}:
$\text{NNT} = 1/(\text{absolute risk reduction})$.

\begin{definition}
The \textbf{Number Needed to Detain (NND)} is $\text{NND} = 1/\text{PPV}$: the expected number of flagged defendants detained for each violent offense prevented (assuming detention prevents the flagged offense).
\end{definition}

This translates model performance into a concrete policy tradeoff: \emph{how many people are detained to prevent one violent offense}. While NNT ranges from 3--100 for effective medical interventions \cite{Laupacis1988}, current NND values of 6--9 mean that for each prevented offense, between 5 and 8 of the detained defendants would not have offended. Table~\ref{tab:nnd} shows NND across base rates.

\begin{table}[t]
\centering

\caption{NND at current instrument performance ($\text{LR} = 4$). NND = $1/\text{PPV}$: defendants detained per prevented offense.}

\label{tab:nnd}
\begin{tabular}{lccc}
\toprule
Base Rate & LR & PPV & NND \\
\midrule
5\% & 4 & 17\% & 6 \\
3\% & 4 & 11\% & 9 \\
2\% & 4 & 8\% & 13 \\
1\% & 4 & 4\% & 25 \\
\bottomrule
\end{tabular}
\end{table}

\subsection{Legal Standards Comparison}

We use PPV targets as \emph{analogical benchmarks} for evidentiary confidence, not legal claims that courts require $\text{PPV} \geq X$.

\paragraph{Why 50\% appears in this paper.}
The 50\% threshold (preponderance of evidence, ``more likely than not'') is a widely understandable reference point. We do not recommend it as a fixed threshold rule for all detention decisions. Pretrial detention involves a distinct legal framework from civil commitment: standards vary by jurisdiction, and dangerousness determinations are not typically expressed as explicit probabilities. The federal Bail Reform Act (18~U.S.C.\ \S~3142) requires a finding that ``no condition or combination of conditions will reasonably assure \dots the safety of any other person and the community,'' with dangerousness established by ``clear and convincing evidence'' \cite{BailReform1984}, a standard typically understood as substantially exceeding 50\%. In a survey of all federal judges, McCauliff \cite{McCauliff1982} found that judges quantified ``preponderance'' at a mean of approximately 55\% and ``clear and convincing'' at approximately 75\%, with low variance across respondents. Stevenson and Mayson \cite{StevensonMayson2022}, in the most sustained analysis of what risk level justifies pretrial detention, concluded through a cost-benefit framework that the threshold required on consequentialist grounds is substantially higher than 50\%, suggesting that our benchmark is conservative rather than demanding. Our contribution is the structural relationship between base rates, LR, and PPV, not one normative cutoff. Readers may apply the framework at any threshold they consider appropriate.

Table~\ref{tab:legal} shows that current ``High Risk for Violence'' designations achieve PPV of 8--17\%, which falls well below a ``more likely than not'' benchmark. We do not claim this is unconstitutional; we note that it reveals a gap between the confidence implied by a ``high-risk'' label and the confidence the data supports.

\begin{table}[t]
\centering

\caption{Legal evidentiary benchmarks (used as analogies, not doctrinal requirements) vs.\ achieved PPV for ``High Violence Risk'' flags. PPV = share of flagged cases that are true positives.}

\label{tab:legal}
\begin{tabular}{lc}
\toprule
Standard & Approximate Probability \\
\midrule
Beyond reasonable doubt & $> 95\%$ \\
Clear and convincing evidence & $> 75\%$ \\
Preponderance (``more likely than not'') & $> 50\%$ \\
\midrule
\multicolumn{2}{l}{\textit{Achieved PPV at current performance:}} \\
$\pi = 5\%$, LR $= 4$ & 17\% \\
$\pi = 3\%$, LR $= 4$ & 11\% \\
$\pi = 2\%$, LR $= 4$ & 8\% \\
\bottomrule
\end{tabular}
\end{table}

\subsection{Uncertainty Labels on Risk Reports}
\label{sec:transparency}

If jurisdictions continue to use these tools, decision-makers should see how often a ``high-risk'' label is correct. We propose that risk outputs used in detention contexts include explicit uncertainty communication:

\begin{center}
\begin{tcolorbox}[colback=red!5!white,colframe=red!75!black,width=0.9\columnwidth,title=\textbf{STATISTICAL CONTEXT}]
\small
\textbf{This ``High Violence Risk'' flag does NOT mean the defendant will be violent.}

At current instrument performance ($\text{LR} \approx 4$, base rate $\approx 3\%$):
\begin{itemize}
    \item This flag is correct approximately \textbf{1 in 9 times} (11\% PPV).
    \item To prevent one violent offense, \textbf{9} flagged defendants would be detained.
    \item This confidence level (11\%) is \textbf{well below} a ``more likely than not'' benchmark ($>$50\%).
\end{itemize}
\end{tcolorbox}
\end{center}

\section{Discussion}

For rare violent outcomes, current pretrial risk tools do not appear to generate enough signal for a ``high-risk'' label to be correct often enough to support high-confidence individualized detention decisions.

\subsection{Surveillance and Differential Precision}

The Surveillance Ceiling further constrains achievable performance for populations subject to differential policing.
When people who would not commit an offense accumulate risk factors through heightened surveillance, the best achievable precision from those features is reduced.
This effect arises from the data itself, not biased modeling.
The same instrument applied uniformly can yield systematically lower precision for over-policed groups even at equal offense rates.
No threshold adjustment or fairness constraint can recover predictive signal lost when recorded features are distorted by unequal surveillance exposure.

\subsection{What Would Clear the Wall?}

Current instruments rely on criminal history and administrative records, extensively studied features unlikely to yield order-of-magnitude LR gains.
The wall is lower for more common outcomes: instruments predicting ``any re-arrest'' (base rates 20--30\%) can achieve substantially higher PPV.
The constraint is specific to \emph{rare} outcomes like violent felony.

\subsection{Why AUC Misleads}

Many validation studies emphasize AUC, which measures how well a tool ranks people. But AUC does not tell us how often a ``high-risk'' flag is correct, and it does not change when the outcome becomes rarer. A tool can have similar AUC across settings while the practical meaning of a ``high-risk'' flag changes dramatically. Table~\ref{tab:auc} illustrates this disconnect.

\begin{table}[t]
\centering

\caption{A model with the same ranking performance (AUC $\approx 0.70$) can have very different PPV depending on how common the outcome is.}

\label{tab:auc}
\begin{tabular}{lcc}
\toprule
AUC & Base Rate & Approximate PPV \\
\midrule
0.70 & 50\% & $\sim$70\% \\
0.70 & 10\% & $\sim$25\% \\
0.70 & 3\% & $\sim$10\% \\
\bottomrule
\end{tabular}
\end{table}

\section{Limitations}

This work has several important limitations.
First, our theoretical results characterize \emph{structural constraints} on PPV under rare-event base rates and threshold-based decisions. They do not imply that risk tools are uniformly inaccurate or that probabilistic modeling lacks value in all contexts. Our results apply specifically to individualized, high-confidence decisions like preventive detention, where PPV is the most relevant quantity.

Second, the Surveillance Ceiling is derived under a stylized model (count-threshold classifiers with independent binary indicators). While this captures the core logic of additive point-score instruments, real-world tools may involve correlated features and nonlinear interactions. As discussed in Section~5, the directional conclusion is likely robust, but quantitative magnitudes are assumption-dependent.

Third, our empirical illustration relies on one publicly available dataset. The results are illustrative rather than nationally representative. Fourth, ``violent recidivism'' is measured via re-arrest, which may conflate behavior with surveillance intensity.

Fifth, we focus on threshold-based metrics. Other evaluative frameworks (ranking-based, population-level) may raise distinct questions. Our results do not rule out predictive models for non-liberty-depriving interventions or aggregate planning.

Sixth, our analysis is static and cross-sectional. It does not model temporal feedback loops in which predictions influence policing, which generates future training data. Modeling such dynamics \cite{Lum2016,Ensign2018} is important future work.

Finally, our base rates are derived from released defendants. Incapacitation effects may raise the ``true'' rate, but empirical evidence suggests these effects are too modest to bridge the LR gap (Section~7).

\section{Ethical Considerations}

This research uses exclusively de-identified, publicly available data.

Our findings connect to broader critiques of automated risk assessment. Green \cite{Green2020} argues that risk assessments embed existing institutional logics and are limited as instruments of reform. Citron and Pasquale \cite{Citron2014} raise due process and contestability concerns about automated predictions affecting liberty. Starr \cite{Starr2014} questions whether group-based statistics can legitimately ground individualized punishment. Our results provide quantitative support: the Likelihood Ratio Wall shows that the confidence carried by ``high-risk'' flags may be far lower than their high-stakes use requires.

Our findings also carry downstream risks. By showing that current features are insufficient for high-precision detention decisions, we risk incentivizing more invasive data collection. We explicitly oppose such expansion. We emphasize transparency about uncertainty rather than expansion of surveillance.

\section{Conclusion}

In settings where violent re-arrest is rare, a ``high-risk for violence'' label can sound far more certain than the data support.
At 2--5\% base rates, 50\% PPV requires $\text{LR} \geq 19$--49.
Current instruments achieve $\text{LR} \approx 2$--6, falling short by factors of 3--25.

This gap is structural.
Recalibration cannot close it (Theorem~\ref{thm:recal}).
The Surveillance Ceiling (Theorem~\ref{thm:ceiling}) compounds the problem: for over-policed populations, differential factor accumulation degrades maximum achievable precision.
Our empirical validation confirms superlinear FPR amplification: a $1.8\times$ factor-count ratio produces a $2.9\times$ FPR ratio.
At current performance, roughly 7--10 defendants are detained per prevented violent offense.
Whether this is acceptable is a policy judgment, but one that should be made with accurate information.

\paragraph{Recommendations.}
\begin{enumerate}
\item Report LR, PPV, and NND at operational thresholds, not just AUC.
\item Do not use ``High Risk for Violence'' labels to support high-confidence detention decisions when PPV is below 20\%; require explicit uncertainty disclosures if such labels are retained.
\item Report performance and NND stratified by demographic group.
\item Include uncertainty communication in outputs presented to judges (Section~\ref{sec:transparency}).
\end{enumerate}

The \textbf{99-to-1 Rule} encapsulates our finding: at a 1\% base rate, coin-flip reliability requires a roughly 99-fold signal advantage.
No evidence suggests such signals exist in currently deployed criminal history data.
Until they are found---or their absence acknowledged---we should not mistake statistical noise for certainty in decisions affecting liberty.

\section*{Generative AI Statement}

The author used AI-assisted tools to assist with grammar, fluency, and LaTeX formatting. No AI system contributed to the formulation of theoretical arguments, proofs, interpretations, or empirical analyses.

\begin{acks}
The author acknowledges the support of the Natural Sciences and Engineering Research Council of Canada (NSERC) [funding reference number RGPIN-2019-04085].
\end{acks}


\bibliographystyle{ACM-Reference-Format}
\bibliography{references}

\appendix

\section{Full Proof of Theorem~\ref{thm:ceiling} (Surveillance Ceiling)}
\label{app:ceiling_proof}

We provide the complete proof expanding the main-text sketch.

Let $S_k = \frac{1}{k}\sum_{i=1}^k X_i$. We assume $S_k$ satisfies a Large Deviations Principle with rate function $I_G(x) = \sup_{\lambda} \{ \lambda x - \Lambda_G(\lambda) \}$, where $\Lambda_G$ is the cumulant generating function for group $G$.

\paragraph{(i) FPR Amplification.}
The false positive rate is $q_G(k) = \Pr(S_k \ge \theta \mid Y=0, G)$. By the G\"artner-Ellis Theorem, $\lim_{k\to\infty} \frac{1}{k} \log q_G(k) = -I_G(\theta)$. Because $p_B > p_A$, we have $I_B(\theta) < I_A(\theta)$ for any $\theta > p_B$:
\[
\lim_{k\to\infty} \frac{1}{k} \log \frac{q_B(k)}{q_A(k)} = I_A(\theta) - I_B(\theta) > 0.
\]
Under i.i.d.\ Bernoulli, $I_G(\theta) = D(\theta \| p_G) = \theta \log \frac{\theta}{p_G} + (1-\theta) \log \frac{1-\theta}{1-p_G}$.
Since $\theta > p_B > p_A$, $D(\theta \| p_A) > D(\theta \| p_B)$, confirming exponential divergence. This holds more generally under dependence structures permitting a differentiable cumulant generating function.

\paragraph{(ii) PPV Ceiling.}
Since sensitivity $s(k)$ is determined by the positive-class distribution (identical across groups), $\text{LR}_G(k) = s(k)/q_G(k)$. The shift $p_A \to p_B$ increases $q_B$ relative to $q_A$, so $\text{LR}_B(k) < \text{LR}_A(k)$. Because PPV is strictly increasing in LR for fixed prior odds (Theorem~\ref{thm:wall}), $\text{PPV}_B < \text{PPV}_A$ at every threshold. Taking the supremum: $\sup_m \text{LR}_B(m) \le \sup_m \text{LR}_A(m)$. \qed

\end{document}